\documentclass[10pt,journal,twocolumn]{IEEEtran}
\usepackage{amsmath,amssymb,amsthm,amsfonts,mathrsfs,fullpage}
\usepackage{tablefootnote}
\usepackage[margin=0.5in]{geometry}

\usepackage{hyperref} 

\usepackage{graphicx,color,cases} 
\usepackage{xfrac}  
\definecolor{DarkGreen}{rgb}{0.1,0.5,0.1}
\definecolor{DarkRed}{rgb}{0.5,0.1,0.1}
\definecolor{DarkBlue}{rgb}{0.1,0.1,0.5} 
 
\newtheorem{theorem}{Theorem} 
\newtheorem{lemma}[theorem]{Lemma} 
\newtheorem{corollary}[theorem]{Corollary}

\theoremstyle{definition}

\numberwithin{equation}{section}
 
\def\>{\rangle} 
\def\<{\langle}
 
 \def\case#1{{\left\{  
	\begin{array}{ll}
  #1
	\end{array}
 \right.   }} 

\DeclareMathOperator{\tr}{Tr}

\begin{document}

\title{Robust quantum metrology with explicit symmetric states}%

\author{ 
    \IEEEauthorblockN{Yingkai Ouyang\IEEEauthorrefmark{1}\IEEEauthorrefmark{2}, Nathan Shettell\IEEEauthorrefmark{3} and Damian Markham \IEEEauthorrefmark{3}}\\
    \IEEEauthorblockA{\IEEEauthorrefmark{1}Department of Physics \& Astronomy, University of Sheffield, Sheffield, S3 7RH, United Kingdom}\\
    \IEEEauthorblockA{\IEEEauthorrefmark{2}Department of Electrical \& Computer Engineering, National University of Singapore, Singapore}\\
    \IEEEauthorblockA{\IEEEauthorrefmark{3}Laboratoire d'Informatique de Paris 6, CNRS, Sorbonne Universit\'e, 4 place Jussieu, 75005 Paris, France}
\thanks{Y. Ouyang is with the National University of Singapore. E-mail: oyingkai@gmail.com}
}

\maketitle

\begin{abstract}
Quantum metrology is a promising practical use case for quantum technologies,
where physical quantities can be measured with unprecedented precision.
In lieu of quantum error correction procedures, near term quantum devices are expected to be noisy, and we have to make do with noisy probe states. We prove that, for a set of carefully chosen symmetric probe states that lie within certain quantum error correction codes, quantum metrology exhibits an advantage over classical metrology even after the probe states are corrupted by a constant number of erasure and dephasing errors. These probe states prove useful for robust metrology not only in the NISQ regime, but also in the asymptotic setting where they achieve Heisenberg scaling. This brings us closer towards making robust quantum metrology a technological reality.
\end{abstract}

\section{Introduction} 
To harness the full powers of quantum technologies, quantum error correction is necessary to mitigate the inevitable decoherence of quantum information. 
However, in lieu of the era of quantum error correction, it would nonetheless be great to be able to unlock some of the potential of quantum technologies in near term quantum devices.
We are approaching the Noisy Intermediate-Scale Quantum (NISQ) era \cite{nisq}, where quantum devices, albeit noisy, will have between 50 to 100 qubits in the near future. 
An important question is what quantum advantage such NISQ devices might offer us in the near-term.
Given that quantum metrology appears to be less demanding of the level of precision that is required to manipulate quantum data as compared to a fully fledged quantum computing device, 
one might wonder if quantum metrology could provide a quantum advantage using these NISQ devices.

The main idea behind quantum metrology is to allow high-resolution and highly sensitive measurements of physical parameters by consuming a quantum resource state, often called the probe state.
We expect that the probe state utilized will be highly entangled, which might still have some entanglement even when noisy and thereby still have some use.
If quantum metrology does provide a quantum advantage in practice, it would allow the development of new sensor chips that utilize quantum entanglement to achieve unprecedented precision and sensitivity in their measurements.
This has implications across all fields where sensors are important, such as in detecting gravitational waves \cite{abbott2016observation}, enhancing radar technologies \cite{maccone2019quantum}, and increased sensitivity in medical measurements \cite{taylor2016quantum},
optical interferometry \cite{PhysRevA.80.013825}, field sensing \cite{unden2016quantum}, Hamiltonian tomography \cite{zhang2014quantum,kura2018finite} and deformation sensing \cite{sidhu2017quantum,sidhu2018quantum}.

While quantum noise degrades the quality of entanglement within a probe state consumed in quantum metrology, it is nonetheless possible to yield a quantum advantage in sensing if (1) the highly entangled probe state is carefully chosen, and (2) if the noise is not too severe.
The scenario of interest is that of robust quantum metrology, which we define to be the following.
A chosen probe state is passively exposed to noise, with no application of quantum error correction.
This is in contrast to other schemes for noisy quantum metrology where quantum error correction is performed multiple times on the probe state before the end of signal accumulation to mitigate the impact of noise accumulation \cite{kessler2014quantum,dur2014improved,arrad2014increasing,demkowicz2014using,unden2016quantum,matsuzaki2017magnetic,Zhou2018,layden2019ancilla,gorecki2019quantum,PRXQuantum.2.010343}.
The corrupted probe state is then directly used for the purpose of quantum metrology.
If the quantum Fisher information (\textsc{QFI}) of the resultant optimal measurement exceeds that of optimal classical Fisher information (\textsc{CFI}), 
we say that the chosen probe state allows for robust quantum metrology.
We emphasize that in contrast to metrology schemes that employ quantum error correction that often consider a number of errors that grows with the system size, we consider here a constant number of errors while the system size grows.

To illustrate our findings, we consider the canonical problem in quantum metrology where the observable to be measured is the spin of a single qubit. 
For the interested reader, we suggest a recent review on the field of quantum metrology \cite{sidhu2017quantum}, which has been extensively studied by a broad community.
By using multiple measurements, classically, the signal quantified by the \textsc{CFI} can be enhanced by a factor of $N$, where $N$ denotes the number of measurements. 
It is well-known that if one prepares an $N$-qubit GHZ state and then measures each qubit identically, the corresponding \textsc{QFI} can reach $N^2$, and thereby greatly surpass the best possible classical strategy. 
Indeed, in the noiseless case, the GHZ state is the optimal probe state.
This quantum advantage becomes especially prominent when $N$ is very large.
This $N^2$ scaling, known as the Heisenberg scaling, however vanishes when there is a single erasure or phase error, as the GHZ state becomes a classical mixture of the all 0s and all 1s state.
  
It has been shown that for i.i.d. noise, the Heisenberg scaling of quantum metrology is lost, and shot noise behavior reflecting classical scaling is recovered \cite{DemkowiczDobrzaski2012}.
We thereby lose the asymptotic quantum advantage of quantum metrology in this setting.
However, there might still be hope for robust metrology in an intermediate noise regime, where the ratio of the number of errors to the number of qubits vanishes asymptotically.
Recently, Oszmaniec {\em et al.} looked into the use of 
random states of distinguishable particles for quantum
metrology when a constant number of particles are erased \cite{PRX-random-states}.
This corresponds to a scenario where a known subset of qubits are damaged. 
Quantum metrology can then be performed on the remaining undamaged qubits.
They found that even if these states are pure and hence typically highly entangled,
they are almost surely useless for metrology.
Remarkably, random symmetric states are almost surely useful under finite particle loss.
This suggests that it would be fruitful to consider using explicit symmetric states for robust quantum metrology.
However this problem is also non-trivial because, as mentioned, the most obvious symmetric state, the GHZ state, is known to be bad for robust quantum metrology, because a single $Z$ error can totally dephase it.
Aside from the GHZ states, other symmetric states such as spin-squeezed states \cite{ma2011quantum} and symmetrized GHZ states \cite{Baumgratz2016_PRL} have also been considered for quantum metrology .

Quantum states that comprise quantum error correcting codes are known to be highly entangled, 
and intuitively, one expects that it is their underlying entanglement that imparts some of their error correction capabilities. 
While one might expect that  quantum states that are good for quantum error correction ought to be also good for robust quantum metrology,
this intuition is false, because while random quantum codes are almost surely good quantum error correction codes \cite{ABKL00,FeM04,MaY08,JiX11,ouyang2014concatenated},
random quantum states are almost surely useless for robust quantum metrology \cite{PRX-random-states}.

In this paper, we study the performance of quantum states that arise from some of these symmetric quantum codes in Ref.~\cite{ouyang2014permutation} for use in quantum metrology.
Permutation-invariant quantum codes are quantum codes that lie entirely within the symmetric space \cite{Rus00,PoR04,ouyang2014permutation,ouyang2015permutation,OUYANG201743,ouyang2019permutation}, and are invariant under any permutation of their underlying particles.  
Research in permutation-invariant codes has not solely been just of theoretical interest, as the preparation of such codes in physically realistic scenarios has been studied \cite{init-picode-2019-PRA}, and have also been considered to use for quantum storage \cite{ouyang-memories}.

Our first result for robust metrology, addresses the scenario of erasure errors.
Here, we provide analytical lower bounds of the \textsc{QFI} for explicitly chosen symmetric probe states. 
Asymptotically. we show that the \textsc{QFI} can attain the Heisenberg scaling with a single or two erasures. 

Our second result is an analytical lower bound on the \textsc{QFI} for our probe state when dephasing errors occur on our qubits. 
By being able to consider dephasing errors, we go beyond the paradigm of Oszmaniec {\em et al.} \cite{PRX-random-states}, where the only type of errors considered for robust metrology are erasure errors.
We consider a noisy channel that comprises of convex combinations of unitary processes where either zero or one phase error occurs uniformly on any of the underlying qubits. For such a noise model, when the probability of zero phase errors is equal to the probability of one phase error, a GHZ state becomes a convex combination of the all zeros state and and the all ones state, and has zero QFI.
In contrast our probe state can have a positive QFI in this scenario (see Theorem \ref{thm:single-qubit-dephasing-non-asymptotic}).
We also discuss how our analysis applies to the scenario of i.i.d. dephasing noise on all qubits, in the limit where the probability of dephasing per qubit approaches zero faster than the reciprocal of the number of qubits.

Interestingly, our lower bounds on the \textsc{QFI} can be expressed in terms of the summation of binomial coefficients $\binom n k$ multiplied by polynomials in the variable $k$. Moreover these polynomials are closely related to Krawtchouk polynomials that commonly arise in classical coding theory. 

We believe that our results pave the way forward towards realizing robust metrology in the NISQ era. This is because as one can see from Figure \ref{fig:bounds} and Figure \ref{fig:bounds dephasing}, a quantum advantage can in principle already be attained using our proposed noisy probe states for either a single erasure error or a single dephasing error. 

  \section{Explicit symmetric probe states}
  In the problem of phase sensing, an unknown parameter $\chi$ of a Hamiltonian $\chi H$ is to be measured. In this paper, we take $H$ to be a sum of phase-flip Pauli operators, so that
\begin{align}
H = \sum_{j=1}^N Z_j,
\end{align}
where $Z_j$ denotes the Pauli operator that applies a $Z$ on qubit $j$ and applies the identity operator on all other qubits.
The \textsc{QFI} can be used to quantify the performance of a quantum state $\rho$ for quantum metrology with respect to the generator $H$. 
It is well known that a lower bound for the \textsc{QFI} can be obtained with the trace norm of commutator of $\rho$ and $H$ \footnote{One can refer to \cite{PRX-random-states} for example}. 
Namely, the \textsc{QFI} is at least
\begin{align}
  \| [\rho, H] \|_1^2    \ge    \| [\rho, H] \|_2^2,
\end{align}
where $[\rho,H] = \rho H - H \rho$ and $\| A\|_p $ denotes the $p$-norm of the vector of singular values of $A$.
Using Eq.~(\ref{eq:norm2-[A,B]}) we find that
\begin{align}
     \| [\rho ,H] \|_2^2 
     &=
     2 \tr (\rho^2 H^2) - 2 \tr( \rho H \rho H ) , \label{eq:QFI bound}
\end{align}
and it is \eqref{eq:QFI bound} that we use to evaluate a generator-type lower bound on the \textsc{QFI} \footnote{See the appendix for details.}.

For reasons explained earlier, we wish to explore metrology on symmetric states.
While certain symmetric states are known to exhibit a large amount of entanglement \cite{Mar11}, there is no guarantee that they are useful for robust quantum metrology. 
The state that we propose to use for robust metrology is some well chosen state within the codespace of a permutation-invariant code that corrects $t$ errors. 
The intuition is that since the quantum code is symmetric and can correct errors, its codewords ought to be useful for robust metrology.
The intuition needs to be quantified, and we achieve this here.

While many families of permutation-invariant codes have been studied  \cite{Rus00,PoR04,ouyang2014permutation,ouyang2015permutation,OUYANG201743,ouyang2019permutation},
we focus our attention on a code family supplied in \cite{ouyang2014permutation} which is completely described by three parameters, given by $g,n$ and $u$. 
Intuitively, $g$ and $n$ are parameters that control the number of correctible bit-flip  and phase-flip errors respectively, and $u$ is a scaling parameter that is at least 1 and $(u-1)$ is the proportion of extra qubits used.
Different choices of $u$ do not decrease the number of errors the symmetric code can correct \cite{ouyang2014permutation}.
The total number of qubits comprising of these codes is $N = g n u$, 
and such codes are known as gnu codes.
For technical reasons, we choose $u=1$ to provide robustness against erasure errors, and we choose $u=2$ to provide robustness against dephasing errors. 
The corresponding logical codewords are
\begin{align}
|0_L\> = \sqrt{2^{-(n-1)}} \sum_{\substack{0 \le j \le n \\ j\ {\rm even}} } 
\sqrt{\binom n j}
|D^{gnu}_{gj}\>	 ,\notag\\
|1_L\> = \sqrt{2^{-(n-1)}} \sum_{\substack{0 \le j \le n \\ j\ {\rm odd}} } 
\sqrt{\binom n j}
|D^{gnu}_{gj}\>.
\end{align}
Here $|D^{gnu}_{gj}\>$ are Dicke states on $N=gnu$ qubits with $gj$ 1's.
To be precise, for every $w = 0,\dots, N$, we have
\begin{align}
|D^N_w\> = \frac{1}{\sqrt{\binom N w}} \sum_{\substack{ x_1,\dots, x_N \in \{0,1\} \\ x_1+\dots + x_N = w }} |x_1\>\otimes \dots \otimes |x_N\>.\label{eq:dicke-defi}
\end{align}
Moreover, the quantum code corrects $t$ arbitrary errors whenever $g,n\ge 2t+1$.
For example, to correct 1 error, we have $t=1$, we can have $g=n=3, u=1$ and $N=9$.
It has also been noted that this code shares many mathematical similarities with the binomial codes recently studied in the context of quantum error correction on a single bosonic mode \cite{BinomialCodes2016}.

For our application to metrology, we do not require the full power of quantum error correction. We restrict our attention to a single symmetric probe state that lies within the codespace, which is given by
\begin{align}
|\varphi_u\> = \frac{|0_L\> + |1_L\>}{\sqrt 2} =  \sqrt{2^{-n}} \sum_{j=0}^n \sqrt{\binom n j }|D^{gnu}_{gj}\>	. \label{eq:resource-state-general}
\end{align}
This symmetric probe state can be interpreted as the logical plus operator a permutation-invariant quantum code with parameters $g,n$ and $u$,
and we call this a gnu probe state. 
To summarize, we believe that studying this family of probe states is advantageous because of their symmetry, their quantum error correction properties, and their simple structure in the Dicke basis.
 
 Our first result is that explicitly constructed states, namely our gnu probe states, can serve as good probe states for robust quantum metrology in the case of erasure errors.
When the number of erasures is not too many, we show that the \textsc{QFI} on the unerased qubits approaches the Heisenberg scaling, and exhibits a quantum advantage in the NISQ regime.
More precisely, to protect against $t$ erasure errors, we choose $u=1$ and set $n=2,3,4,5$ and 6, and increase the values of $g$. 
In this scenario, the \textsc{QFI} is asymptotically lower bounded by a quadratic function in $N$, which reproduces the Heisenberg scaling up to a constant.
We present this result formally in Theorem \ref{thm:erasures-exact}, and illustrate our lower bound on the \textsc{QFI} for a number of qubits compatible with the NISQ regime in Figure \ref{fig:bounds}.

In our second result, we use a gnu probe state with $u=2$ and $N=2gn$ qubits.
We consider first the problem of a noisy process that introduces either no phase error or a single phase error randomly on the underlying qubits. 
We calculate an analytical lower bound for the \textsc{QFI} using our probe state in this scenario, and this is given explicitly in Theorem \ref{fig:bounds dephasing}.
We also take the asymptotic limit of large $g$ and constant $n$, with $n \ge 2$. 
In this scenario, we can see from Corollary \ref{coro:dephasing-analytic} that we do recover the Heisenberg scaling.
We can also consider the effect of our probe state against i.i.d type dephasing errors where the probability of dephasing per qubit is $t/(2gn)$.
By only considering the leading order phase errors from this i.i.d dephasing model, we obtain lower bounds for the \textsc{QFI} of our corrupted probe state in this scenario, illustrate our lower bound on the \textsc{QFI} numerically in Figure \ref{fig:bounds dephasing}.

\section{Erasure errors}
In this section, we prove that quantum metrology performed on an explicitly chosen symmetric probe state can recover Heisenberg scaling when very few erasure errors have occured.
For example, when $t$ qubits are known to have been erasured, we can perform metrology on the $N-t$ qubits where no erasures have occurred. 
Here, because of the symmetry of the probe state $|\varphi_1\>$, we may assume without loss of any generality that the erasures always occur on the first $t$ qubits.

To understand what exactly happens when $t$ qubits have been erased from our symmetric probe state $|\varphi_1\>$, we leverage on our ability to explicitly calculate what its corresponding density matrix is when $t$ qubits have been erased. 
This enables us to obtain an analytical lower bound on the \textsc{QFI} when $t$ qubits are lost.

We denote the density matrix for these $N-t$ qubits as $  \rho =\tr_t(|\varphi_1\>\<\varphi_1|),$
where $\tr_t(\cdot)$ denote the partial tracing of the first $t$ qubits. 
We start with a representation of $  \tr_t(|\varphi_1\>\<\varphi_1|)$ in terms of the vectors $ |\theta_0\> , \dots ,  |\theta_t\> $, where
\begin{align}
    |\theta_u\>  = \sum_{j=1} ^{n-1} \frac{ \sqrt{\binom n j} }{\sqrt {2^n} } |H^{N-t}_{gj-u}\> / \sqrt{\binom N{gj} }
\end{align}
and
$|H^{N}_{w} \>  = \sqrt{\binom {N}{w} } |D^{N}_{w} \>$
denotes a rescaling of the Dicke states such that each of its computational basis vector has unit amplitude.
For example, the vectors $|\theta_0\>, \dots, |\theta_t\>$ are pairwise orthogonal, but not orthonormal as seen from the following orthogonality relationship for $u,v = 0, \dots, t$.
\begin{align}
    \<\theta_u |\theta_v\>
    = \delta_{u,v}
     \sum_{j=1} ^{n-1}
     a_{j,u},
     \label{eq:psi-ortho}
\end{align}
where
\begin{align}
a_{j,u} = 
       \frac{ \binom n j }{ 2^n } 
     \frac{ \binom{N-t}{gj-u} } {\binom N{gj} } . \label{eq:aju defi}
\end{align}
Then we have the following lemma.
\begin{lemma}[Representation of the partial trace of our symmetric probe state]
\label{lem:erasure-representation}
Let $g$ and $n$ be positive integers and let $N=gn$. With these parameters, let $|\varphi_1\>$ denote our $N$-qubit symmetric probe state as defined in (\ref{eq:resource-state-general}). 
Let $t$ denote the number of erased qubits in $|\varphi_1\>\<\varphi_1|$, and $  \tr_t(|\varphi_1\>\<\varphi_1|)$ denote the corresponding density matrix obtained.
Then for all $t < g,n$, we have
\begin{align}
    &\tr_t(|\varphi_1\>\<\varphi_1|) \notag \\
    =& 
  \frac{1}{2^{n/2}}  \left(   |H^{N-t}_{0}\>\<\theta_0|  
+       |\theta_0\>\<H^{N-t}_{0}|  
  +      |H^{N-t}_{gn-t}\>\<\theta_t|  
    +   |\theta_t\>\<H^{N-t}_{gn-t}|\right) 
    \notag \\
    &+    
      \sum_{u=0}^t  \binom t u   |\theta_u\>\<\theta_u|
  +    \frac{1}{2^n} \left(  |H^{N-t}_{0}\>\<H^{N-t}_{0}|  +   |H^{N-t}_{gn-t}\>\<H^{N-t}_{gn-t}|  \right).\notag
\end{align}
\end{lemma}
\begin{proof}
First note that
\begin{align}
    |\varphi_1\>\<\varphi_1| &= 
    \sum_{j,k=0}^{n}\frac{\sqrt{\binom n j \binom n k}}{2^n} |D^N_{gj}\>\<D^N_{gk}|   
    \notag\\
    &=    \sum_{j,k=0}^{n}\frac{\sqrt{\binom n j \binom n k }}{2^n \sqrt{ \binom N{gj} \binom N{gk} }}|H^N_{gj}\>\<H^N_{gk}| .\notag
\end{align}
When $1\le j,k \le n-1$,
\begin{align}
    |H^N_{gj}\>\<H^N_{gk}| &= 
    \sum_{u,v=0}^{t} 
    \left(|H^{t}_{u}\> \otimes |H^{N-t}_{gj-u}\>\right)
    \left(\<H^{t}_{v}| \otimes \<H^{N-t}_{gk-v}|\right)    .\notag
\end{align}
By performing the partial trace in the Dicke basis and using the orthogonality condition $\<D^{t}_{u}|H^{t}_{v}\> = \delta_{u,v} \sqrt{\binom t u} $, we get for  $1\le j,k \le n-1$ that
\begin{align}
        \tr_t \left( |H^N_{gj}\>\<H^N_{gk}| \right) &= 
    \sum_{u=0}^t 
 \binom t u    |H^{N-t}_{gj-u}\>  \<H^{N-t}_{gk-u}|. 
\end{align}
Also note that for $j=1,\dots, n-1$, we have
\begin{align}
        \tr_t \left( |H^N_{0}\>\<H^N_{gn}| \right) &=        \tr_t \left( |H^N_{gn}\>\<H^N_{0}| \right)  = 0\notag\\
        \tr_t \left( |H^N_{0}\>\<H^N_{g j}| \right) &=         |H^{N-t}_{0}\>\<H^{N-t}_{gj}|  \notag\\
        \tr_t \left( |H^N_{gj}\>\<H^N_{0}| \right) &=         |H^{N-t}_{gj}\>\<H^{N-t}_{0}| \notag \\
        \tr_t \left( |H^N_{gn}\>\<H^N_{g j}| \right) &=         |H^{N-t}_{gn-t}\>\<H^{N-t}_{gj-t}|\notag  \\
        \tr_t \left( |H^N_{gj}\>\<H^N_{gn}| \right) &=         |H^{N-t}_{gj-t}\>\<H^{N-t}_{gn-t}|\notag\\
        \tr_t \left( |H^N_{0}\>\<H^N_{0}| \right) &=         |H^{N-t}_{0}\>\<H^{N-t}_{0}|\notag  \\
        \tr_t \left( |H^N_{gn}\>\<H^N_{gn}| \right) &=         |H^{N-t}_{gn-t}\>\<H^{N-t}_{gn-t}|  .\notag
\end{align}
Hence 
\begin{align}
   & \tr_t(|\varphi_1\>\<\varphi_1|) \notag\\
    =& 
\frac{1}{2^n}   \sum_{j=1}^{n-1} \sqrt { \frac{\binom n j}{\binom N {gj} } } \left(   |H^{N-t}_{0}\>\<H^{N-t}_{gj}|  
+       |H^{N-t}_{gj}\>\<H^{N-t}_{0}|  
  \right.\notag\\
  &\quad+\left.|H^{N-t}_{gn-t}\>\<H^{N-t}_{gj-t}|  
    +   |H^{N-t}_{gj-t}\>\<H^{N-t}_{gn-t}|\right) 
    \notag \\
    &+    
\frac{1}{2^n}    \sum_{j,k=1}^{n-1}
    \frac{\sqrt{\binom n j \binom n k }}{ \sqrt{ \binom N{gj} \binom N{gk} }}
   \sum_{u=0}^t  \binom t u    |H^{N-t}_{gj-u}\>  \<H^{N-t}_{gk-u}|
   \notag\\
   &
   +
   \frac{1}{2^n} \left(  |H^{N-t}_{0}\>\<H^{N-t}_{0}|  +   |H^{N-t}_{gn-t}\>\<H^{N-t}_{gn-t}|  \right),\notag
\end{align}
from which the result follows.
\end{proof}

\begin{figure}[!htb]
\includegraphics[width=1\columnwidth]{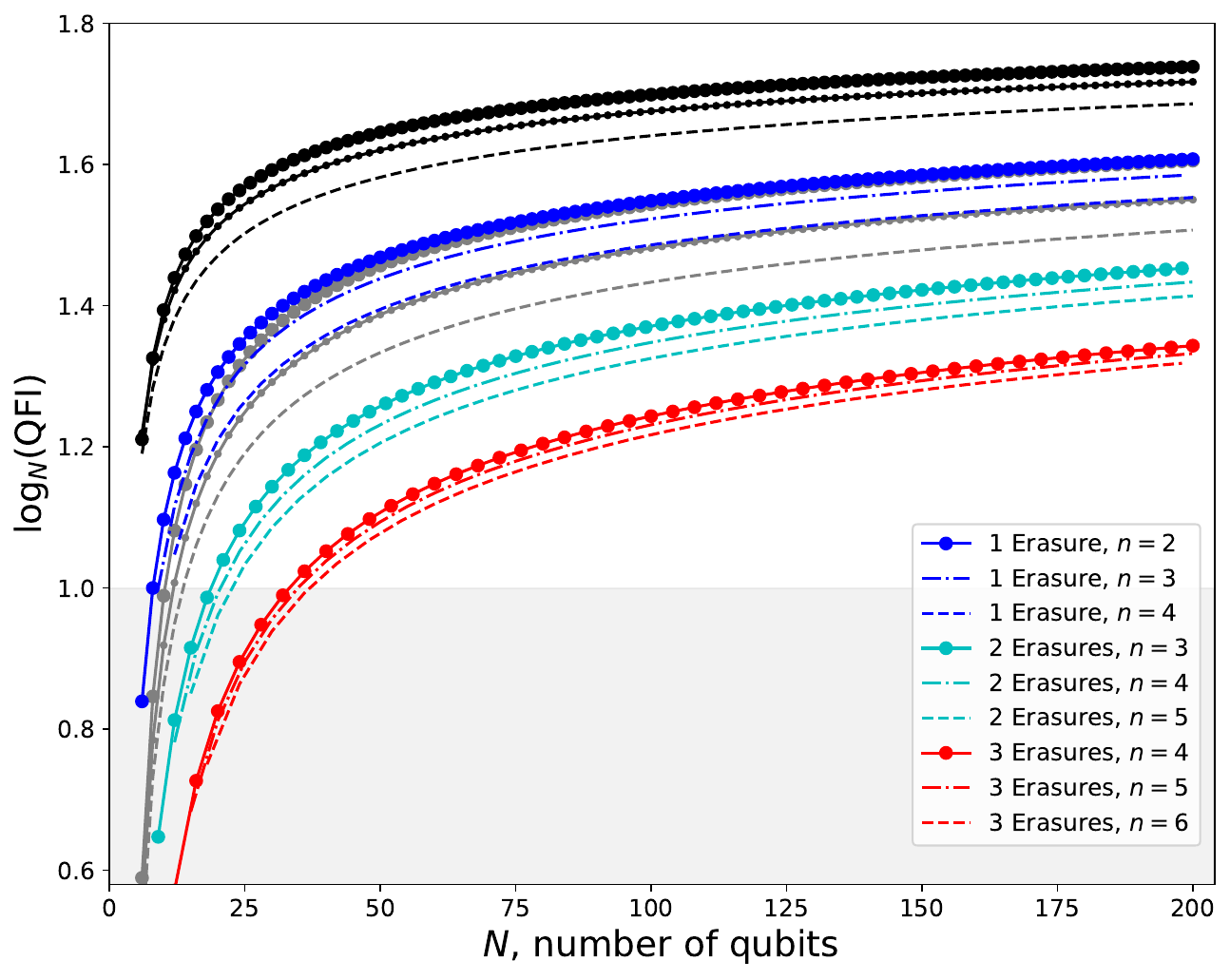}
\caption{Lower bounds for the \textsc{QFI} using the gnu probe state $|\varphi_1\>$ after 1,2 and 3 erasure errors for different values of $n$. When the lines are above 1, there is a quantum advantage.
The colored lines give non-asymptotic lower bounds for the QFI after $t$ erasures using Theorem \ref{thm:erasures-exact}.
The black and grey lines give asymptotic lower bounds for the QFI after $t$ erasures using Theorem \ref{thm:erasure-small-t}.
The black and grey lines correspond to $t=1$ and $t=2$ erasures respectively. 
The topmost black line, the middle black line and the bottom black line correspond to $n=2$, $n=3$ and $n=4$ respectively.
The topmost grey line, the middle grey line and the bottom grey line correspond to $n=3$, $n=4$ and $n=5$ respectively.
We show results for up to 200 qubits, which is the regime of interest in the NISQ era.
When roughly 100 qubits are available, and the rate of erasure is less than 2 percent, there is a discernible quantum advantage in using our symmetric probe state.
Note that the asymptotic lower bounds for $t=2$ match well with the non-asymptotic lower bounds for $t=1$.
}
\label{fig:bounds}
\end{figure}

Now we are in a position to calculate an analytical lower bound on the \textsc{QFI} given by $  \| [\rho ,H] \|_2^2 $. 
The uncorrupted probe state is the pure state $|\varphi_1\>\<\varphi_1|$, 
and when $t$ qubits are erased, the resultant probe state is the density matrix $\rho  = \tr_t(|\varphi_1\>\<\varphi_1|)$.
The generator corresponds only to the unerased qubits, 
and is given by $H=  \sum_{i=1}^{N-t} Z_i$.

Now let 
\begin{align}
   A_u &=  \sum_{j=1}^{n-1} a_{j,u} \notag\\
   B_u &=  \sum_{j=1}^{n-1} a_{j,u}b_{j,u} \notag\\
   C_u &=  \sum_{j=1}^{n-1} a_{j,u}c_{j,u}  \label{eq:ABC},
\end{align}
and $a_{j,u}$ is as given in (\ref{eq:aju defi}), and 
\begin{align}
b_{j,u} &=  \frac{4(gj-u)(N-t-gj+u)}{(N-t)(N-t-1)} ,
\notag\\
c_{j,u} &=  \frac{2(gj-u)}{N-t}.\label{eq:abc}
\end{align}
Now, we present our result for the generator lower bound to the \textsc{QFI} after our gnu probe state has suffered $t$ erasures.
\begin{theorem}
\label{thm:erasures-exact}[$t$ erasures]
Let $g,n$ be positive integers, $N=gn$ and let $t<g,n$. 
Let $|\varphi_1\>$ as defined in (\ref{eq:resource-state-general}) denote our gnu probe state. Let $\rho= \tr_t(|\varphi_1\>\<\varphi_1|)$. 
Then the QFI of $\rho$ is at least 
\begin{align}
               &
\frac{2(N-t)^2}{2^n}  \left(  4A_{t} -B_0   - B_{t} +  2C_{0} -2 C_{t}   \right)\notag\\
&
+
\frac{2(N-t)}{2^n}  \left(  B_{0} + B_{t} \right)
             +
2(N-t)^2  \sum_{u=0}^t  \binom t u^2 K_u,
\label{eq:main technical result for erasures}
\end{align}
where
\begin{align}
K_u &= \left( 2 A_u C_u-\left( 1 - \frac{1}{N-t}\right)A_u B_u - C_u^2 \right).
\end{align}
\end{theorem}
The proof of this while straightforward, is a very tedious calculation. We supply the full details in Appendix \ref{app:A}.
Theorem \ref{thm:erasures-exact} gives a lower bound on the quantum Fisher information when some erasures have occurred on our symmetric probe state, and this lower bound is easy to evaluate.

Armed with an expression which yields a lower bound for the quantum Fisher information under erasures
in Theorem \ref{thm:erasures-exact}, 
we can find an asymptotical lower bound for the quantum Fisher information, in the limit when the number of qubits $N=gn$ that make up our permutation-invariant probe state becomes arbitrarily large.  

In order for our asymptotic analysis to work, the number of erasures $t$ and the number of levels $n$ is constant, while $g$ is taken to be arbitrarily large.

We begin with calculations of the asymptotics of the parameters $a_{j,u}, b_{j,u}$ and $c_{j,u}$ in the limit of large $g$.
\begin{lemma}\label{lem:parameter-asymptotics}
Let $t$ be a positive integer. Then for all $u = 0,\dots, t$ and $j = 1,\dots, n-1$ we have
\begin{align}
 \bar a_{j,u} =  \lim_{g \to \infty} a_{j,u} 
    &= \binom n j 2^{-n} 
    \left( 1 - \frac j n \right)^{t-u} \left(\frac j n \right)^u  \label{eq:a-limit}
 \\
\bar   b_j =  \lim_{g \to \infty} b_{j,u} 
    &= 
    \frac {4j} n \left( 1 - \frac j n \right) \label{eq:b-limit}
 \\
\bar c_j    = \lim_{g \to \infty} c_{j,u} 
    &=  \frac {2j} n . \label{eq:c-limit}
    \end{align}
\end{lemma}
\begin{proof}
Now recall that 
$ a_{j,u}  = \binom n j 2^{-n} \binom {N-t}{gj-u}/\binom N {gj}$.
Since $0 \le u \le t$ and $u \le g$, we have 
\begin{align}
    \binom {N-t}{gj-u}/\binom N {gj}
    =
    \frac{(N-gj)_{t-u}}{(N)_{t}}(gj)_u,\notag
\end{align}
where $(N)_t = N(N-1) \dots (N-t+1)$ denotes the falling factorial.
It then follows that 
\begin{align}
    \lim_{n\to \infty }\binom {N-t}{gj-u}/\binom N {gj}
&=
    \frac{(N-gj)^{t-u}}{N^{t}} (gj)^u \notag\\
&=
    \left(1 - \frac j n \right)^{t-u} \left( \frac j n \right)^u,\notag
\end{align}
from which (\ref{eq:a-limit}) follows.

Now recall that $c_{j,u} = \frac{2(gj-u)}{N-t}$. Taking the limit of large $n$,
we get $\lim_{n \to \infty } c_{j,u} = \frac{2(gj)}{N} = \frac {2j}{n}.$
We similarly get the result for $\bar b_{j }$.
\end{proof}
From the above lemma, we obtain the asymptotic lower bounds on the quantum Fisher information when the number of erasures is small. We present results for a single erasure and two erasures explicitly in Theorem \ref{thm:erasure-small-t}, and it can be readily seen that for any constant number of erasures, we can similarly obtain lower bounds for the \textsc{QFI}.
\begin{theorem}[Asymptotics for one and two erasures]
\label{thm:erasure-small-t}
Let $t$ denote the number of erasures. Let $g$ be arbitrarily large, $n$ be constant. Define our metrological probe state to be based on the parameters $g$ and $n$ as defined in (\ref{eq:resource-state-general}). Then whenever $n>t$, the \textsc{QFI} with $t=1,2$ erasures is at least 
\begin{align}
\case{
 \frac{n-1}{n^2} (N^2  + (n-2)N - (n-1) ) & ,t=1 \\
 \frac{(N-2) (n-1) \left(N \left(3 n^2-n+6\right)+3 n^3-4 n^2+3 n-18\right)}{8 n^4} &,t=2 \\
}.\notag
\end{align}
\end{theorem}
\begin{proof}
Evaluating the limits of $A_j,B_j,C_j$ for $j=0,1,2$, we get using Lemma \ref{lem:parameter-asymptotics} that
\begin{align}
\lim_{g  \to \infty}A_0 &=  \sum_{j=1}^{n-1} \binom n j 2^{-n} (1-j/n) = \frac 1 2 - 2^{-n}\notag\\
\lim_{g  \to \infty}A_1 &=  \sum_{j=1}^{n-1} \binom n j 2^{-n} j/n = \frac 1 2 + \frac{1}{2n} \notag\\
\lim_{g  \to \infty}B_0 &=  \sum_{j=1}^{n-1} \binom n j 2^{-n} (1-j/n) (4j/n)(1-j/n) = \frac {1} {2} -  \frac{1}{2n}  \notag\\
\lim_{g  \to \infty}B_1 &=  \sum_{j=1}^{n-1} \binom n j 2^{-n} (j/n) (4j/n)(1-j/n) = \frac 1 2 -  \frac{1}{2n}  \notag\\
\lim_{g  \to \infty}C_0 &=  \sum_{j=1}^{n-1} \binom n j 2^{-n} (1-j/n)(2j/n) = \frac 1 2 -  \frac{1}{2n} \notag\\
\lim_{g  \to \infty}C_1 &=  \sum_{j=1}^{n-1} \binom n j 2^{-n} j/n(2j/n) = \frac 1 2 +  \frac{1}{2n} - \frac{1}{2^{n-1}} .\notag
\end{align}
and 
\begin{align} 
\lim_{g  \to \infty}A_2 &=  \sum_{j=1}^{n-1} \binom n j 2^{-n}  j^2/n^2 = \frac 1 4 + \frac{1}{4n} - 2^{-n} \notag\\ 
\lim_{g  \to \infty}B_2 &=  \sum_{j=1}^{n-1} \binom n j 2^{-n}  (j/n)^2 (4j/n)(1-j/n)  \notag\\ 
\lim_{g  \to \infty}C_2 &=  \sum_{j=1}^{n-1} \binom n j 2^{-n} j^2/n^2(2j/n) = \frac 1 4 + \frac{3}{4n} - 2^{-(n-1)}.\notag
\end{align}

Hence it follows using Theorem \ref{thm:erasures-exact} that for $t=1$, we have
\begin{align}
\lim_{g \to \infty }\left ( \tr (\rho^2 H^2) -   \tr( \rho H \rho H )  \right) 
\ge   
\frac{(N-1) (n-1) (N+n-1)}{2 n^2}.\notag
\end{align}
Twice of the above gives us a lower bound for the \textsc{QFI} according to (\ref{eq:QFI bound}). This gives us the result for $t=1$.
When $t=2$, we can use Lemma \ref{lem:parameter-asymptotics} to get

From this, for $t=2$, we get
\begin{align}
&\lim_{g \to \infty }\left ( \tr (\rho^2 H^2) -   \tr( \rho H \rho H )  \right)  \notag\\
\ge   &
\frac{(N-2) (n-1) \left(N \left(3 n^2-n+6\right)+3 n^3-4 n^2+3 n-18\right)}{16 n^4} \notag 
\end{align}
Again, twice of the above gives us a lower bound for the \textsc{QFI} according to (\ref{eq:QFI bound}). 
\end{proof}
Calculations for the \textsc{QFI} for larger values of $t$ get increasingly tedious, and we do not pursue this further here.

To interpret the results of Theorem \ref{thm:erasure-small-t} more explicitly, notice that when 
$t=1$ and $g$ is large, the \textsc{QFI} is at least 
\begin{align} 
\case{
\frac{1}{4}(N^2-1) \sim 0.25 N^2, & \quad n = 2\\
\frac{2}{9}(N^2+N-2) \sim 0.22 N^2, & \quad n = 3\\
\frac{3}{16}(N^2+2N-3) \sim 0.19 N^2, & \quad n = 4\\
\frac{4}{25}(N^2+3N-4) \sim 0.16 N^2, & \quad n = 5\\
}.
\end{align}
The results of Theorem \ref{thm:erasure-small-t} suggests that when $n$ is constant and when $g$ is large, 
provided that the number of erasures is strictly less than $n$, the \textsc{QFI} is lower bounded by a constant multiplied by $N^{2}$ which achieves a Heisenberg scaling. 
Indeed, one can show this by performing appropriate leading order analysis on the first term in (\ref{eq:main technical result for erasures}) on Theorem \ref{thm:erasures-exact}. 
 
We evaluate this numerically with $t=1,2,3$, for constant $n$ and increasing $g$ in Figure \ref{fig:bounds}. 
The probe state utilizes a maximum of 200 qubits in the plots.
When the number of qubits is greater than 25, there is a quantum advantage for metrology robust against 2 erasure errors.

\section{Dicke inner products, Krawtchouk polynomials, and binomial summations}
The purpose of this section
is to evaluate in advance quantities that will help us to bound analytical lower bounds on the \textsc{QFI} of our probe state with dephasing errors.
This is because to evaluate a lower bound based on \eqref{eq:QFI bound}, 
it suffices to understand the structure of various Dicke inner products.
To this end, we present the main result in this section, which gives an analytical form for the quantities
\begin{align}
v_1&=\<\varphi_2|Z_1|\varphi_2\>,\notag\\ 
v_2&=\<\varphi_2|Z_1 Z_2 |\varphi_2\>, \notag\\
v_3&=\<\varphi_2|Z_1 Z_2 Z_3|\varphi_2\>, \notag\\
v_4&=\<\varphi_2|Z_1 Z_2 Z_3  Z_4|\varphi_2\> ,
\label{eq:v-defi}
\end{align}
where $|\varphi_2\>$ is the probe state as given in (\ref{eq:resource-state-general}) with $u = 2$ and positive integer parameters $g$ and $n$.
Our result is the following.
\begin{lemma} \label{lem:probe-state-Z-inner-product}
Let $v_1,\dots ,v_4$ be as defined in \eqref{eq:v-defi}. Then 
\begin{align}
v_1 &= \frac{1}{2},\notag\\
v_2 &= \frac{g( n+1)-2}{4 g n-2} = \frac{1}{4} +  \frac{g-3/2}{4gn -2},\notag\\
v_3 &= \frac{g^2 n (n+3)-6 gn+2}{4 (g n-1) (2 g n-1)},\notag\\
v_4 &= \frac{v_{4,3} g^3 -v_{4,2} g^2 + v_{4,1} g - 12}{4 (g n-1) (2 g n-3) (2 g n-1)},
\end{align}
where 
$v_{4,3}=n^3+6 n^2+3 n-2$,
$v_{4,2}=12 n (n+1)$, and
$v_{4,1}=4(5n+2)$.
Furthermore, in the limit of large $g$, for $j=1,2,3,4$, we have 
\begin{align}
\lim_{g\to \infty } v_j &=  \frac{1}{2^j} + \gamma_j.
\label{eq:vs-glimit}
\end{align} 
where
$\gamma_2 = 1/(4n), \gamma_3 = 3/(8n)$ and $\gamma_4 = (6n^2+3n-2)/(16n^2).$
\end{lemma}
Lemma \ref{lem:probe-state-Z-inner-product} 
can be shown once one understands how Dicke inner products when evaluated on Pauli operators can be expressed in terms of Krawtchouk polynomials.
We denote a binary Krawtchouk polynomial by
\begin{align}
K^N_k(z) =  \sum_{j = 0}^z \binom {z} {j} \binom{N-z}{k-j} (-1)^j. \label{eq:kpoly summation}
\end{align}
In the language of generating functions, 
\begin{align}
    K^N_k(z) = [x^k](1-x)^{N-z}(1+x)^z,
\end{align}
where $[x^k]f(x)$ denotes the coefficient of $x^k$ of a polynomial $f(x)$. 
Recall that a Dicke state is a normalized superposition on $m$ qubits of all permutations of computation basis vectors with $w$ 1s and $m-w$ 0s, by $|D^N_w\>$, which we have defined in \eqref{eq:dicke-defi}. 
Then, it can be seen that Dicke inner products of Paulis can be evaluated with the following lemma.
\begin{lemma}
\label{lem:dicke-inner}
Let $N$ be a positive integer, and $w,a$ be non-negative integers such that $w+a\le N$. 
Let $x,y,z$ be non-negative integers such that $x+y+z \le N$.
Let $P$ be any Pauli operator with $x$ $X$'s, $y$ $Y$'s and $z$ $Z$'s.
Then if $x+y-a$ is odd, $\<D^N_{w+a} | P |  D^N_w\> = 0$.
If $x+y-a$ is even, then 
\begin{align}
\<D^N_{w+a} | P |  D^N_w\>  
&= \frac{i^y}{\sqrt{ \binom N w \binom N {w+a}  }}
 K^{x+y}_{\frac{x+y-a}{2}}(y) K^{N-x-y}_{w-\frac{x+y-a}{2}}(z) .\notag
\end{align} 
\end{lemma}
\begin{proof}
Let $|\theta_t\>$ be a computation basis vector on $N$ qubits with $w$ $|1\>$s and $N-w$ $|0\>$s.
Let $n_x,n_y,n_z$ denote the number of $|1\>$s $|\theta_t\>$ has on the support of the $X$, $Y$ and $Z$ part of $P$ respectively.
Clearly $0 \le n_x \le x$, $0 \le n_y \le y$, $0 \le n_z \le z$.
Moreover we are interested in $n_x+n_y+n_z \le w$. 
Clearly $P|\theta_t\>$ is up to a phase also a computation basis vector $|\phi\>$, and this phase is equal to $i^y (-1)^{n_y+n_z}$. 

We now proceed to count the number of $|1\>$s in the computation basis vector $|\phi\>$. Let this number be $n$. Then we have
\begin{align}
n &= (x-n_x) + (y-n_y) + n_z + (w-n_x-n_y-n_z) \notag\\
&=w+x+y-2(n_x+n_y).
\end{align}
If $n=w+a$, then we must have
\begin{align}
w+a &=w+x+y-2(n_x+n_y) \notag\\
2(n_x+n_y) &= x+y -a, 
\end{align}
and this implies that $x+y-a$ must be even.
This implies that if $x+y-a$ is odd, $\<D^N_{w+a} | P |  D^N_w\>  $ must be zero.

Now we can evaluate $\<D^N_{w+a} | P |  D^N_w\>  $ when $x+y-a$ is even.
Then we need only take into account all the ways of picking $|\theta_t\>$ for $P$ to act on, and for each of these instances, sum the appropriate phase. 
Now notice that $n_x = \frac{x+y-a}{2} - n_y$. 
Thus,
\begin{align}
\<D^N_{w+a} | P |  D^N_w\> 
=
\frac{1}{\sqrt{\binom N w \binom N {w+a} }}
\sum_{\substack{
n_x , n_y,n_z \ge 0\\
n_x+n_y+n_z \le w \\
n_x \le x \\
n_y \le y\\
n_z \le z\\
}}
i^y (-1)^{n_y+n_z}.
\end{align}
Expressing this in terms of binomial coefficients, 
and adopting the convention that 
$\binom n k = 0$ for negative $k$, 
we have
\begin{align}
&\<D^N_{w+a} | P |  D^N_w\>  \notag\\
=&
\frac{i^y}{\sqrt{\binom N w \binom N {w+a} }} 
\sum_{n_z=0}^z \binom{z}{n_z} \binom{N-x-y-z}{w-n_x-n_y-n_z}
\notag\\
&\times
(-1)^{n_z}
\sum_{n_y=0}^y \binom{y}{n_y} \binom{x}{ \frac{x+y-a}{2} - n_y}
 (-1)^{n_y} \notag\\
 =&
 \frac{i^y}{\sqrt{\binom N w \binom N {w+a} }} 
K^{N-x-y}_{w-n_x-n_y}(z)
K^{x+y}_{\frac{x+y-a}{2}}(y) .
\end{align}
Substituting $n_x+n_y = \frac{x+y-a}{2}$ then yields the result.
\end{proof}
 For the purpose of proving Lemma \ref{lem:probe-state-Z-inner-product}, it suffices to consider only Krawtchouk polynomials
 $K^N_w(z)$ for $z=0,1,2,3,4$.
Using Lemma \ref{lem:dicke-inner}, it is easy to see that 
\begin{align}
     \< D^N_{w}| Z_1\dots Z_z  |D^N_{w}\> &= K^N_w(z) / \binom N w. 
\end{align}
When $z \ge w$ in the expression $ K^N_w(z)$ and $ w\le m-z$, we use 
(\ref{eq:kpoly summation}) to express $\binom m w^{-1} K^N_w(z) $ as polynomials in $w$.
From this, we find that 
\begin{align}
 \< D^N_{w}| Z_1  |D^N_{w}\> &= 1-\frac{2w}{N}.\label{eq:dickeZ1} \\
  \< D^N_{w}| Z_1 Z_2 |D^N_{w}\> &=
  1 +   \frac{4w}{N-1} -  \frac{4w^2}{N(N-1)}.  
 \label{eq:dickeZ2}
  \end{align}
Similarly, we find
  \begin{align}
&    \< D^N_{w}| Z_1 Z_2 Z_3 |D^N_{w}\> \notag\\
  =&
1 + w\frac{-6 N^2+6 N-4}{N \left(N^2-3 N+2\right)} + \frac{12w^2}{N^2-3 N+2}\notag\\&
-\frac{8w^3}{N \left(N^2-3 N+2\right)},
 \label{eq:dickeZ3}
 \end{align}
 and denoting $(N-1)_3 = (N-1)(N-2)(N-3),$ $(N)_4 = N(N-1)(N-2)(N-3)$ as falling factorials we get
 \begin{align}
 & \< D^N_{w}| Z_1 Z_2 Z_3 Z_4 |D^N_{w}\>
  \notag\\
=&
  1-  \frac{8w \left(N^2-3 N+4\right)}{(N-1)_3} + \frac{8 w^2\left(3 N^2-3 N+4\right)}{(N)_4}
  \notag\\
& - \frac{32w^3}{(N-1)_3}+\frac{16w^4}{(N)_4}.
 \label{eq:dickeZ4}
\end{align} 
Given the identities \eqref{eq:dickeZ1}, \eqref{eq:dickeZ2}, \eqref{eq:dickeZ3} and \eqref{eq:dickeZ4}, we are in a position to prove Lemma \ref{lem:probe-state-Z-inner-product}.
\begin{proof}[Proof of Lemma \ref{lem:probe-state-Z-inner-product}]
It is clear that 
\begin{align}
v_1 &= \sum_{j=0}^n 2^{-n} \binom n j \<D^N_{gj}| Z_1 |D^N_{gj}\> \label{eq:expand-1}.
\end{align}
By substituting \eqref{eq:dickeZ1} into \eqref{eq:expand-1},
we find that 
\begin{align}
v_1 &= \sum_{j=0}^n 2^{-n} \binom n j \left( 1 - 2(gj)/N \right)	 \label{eq:expand-2}.
\end{align}
By using the binomial identities $\sum_{j=0}^n \binom n j = 2^n$ and 
$\sum_{j=0}^n \binom n j j = 2^{n-1}n $, and $N=2gn$, we get
\begin{align}
v_1 = 1 - gn / (2gn) = 1/2,
\end{align}
which completees the proof that $v_1 = 1/2$.
The proof of the remaining results can be found using a similar methodology, and using more binomial identities for evaluating $\sum_{j=0}^n \binom n j j^x$ for $x=0,1,2,3,4$. The asymptotic results for large $g$ can be found directly by taking limits of the analytical expressions we find for $v_1,v_2,v_3$ and $v_4$.
\end{proof}

\section{Dephasing errors}
An archetypal scenario of dephasing on multiple qubits is one where every qubit dephases with probability $p$,
which means that independently on every qubit, no error or a Pauli $Z$ error applies on each qubit with probability $1-p$ and $p$ respectively.
To allow robust metrology on using probe states that are allowed to accumulate a limited amount of dephasing errors, 
we propose to use an initial probe state $|\varphi_2\>$ given in (\ref{eq:resource-state-general}) where $u =2$ so that the total number of qubits used is $N = 2gn$.
This proposed probe state uses twice as many as that used for erasure errors for fixed $g$ and $n$.
Explictly, the dephased $N$-qubit probe state is 
\begin{align}
  \sigma &= 
    \sum_{j=0}^N   p^j (1-p)^{N-j}
    \sum_{       {\bf x}  \in B_j  } 
    Z_{{\bf x}}     |\varphi_2\>\<\varphi_2|     Z_{{\bf x}}  
    \label{eq:dephasing definition},
\end{align}
where $p$ denotes the probability that each qubit is dephased, $B_j$ denotes the set of all $N$-bit binary vectors with $j$ 1s, and 
$Z_{{\bf x}}=  Z^{x_1} \otimes \dots \otimes Z^{x_N}$ for every binary vector ${\bf x} = (x_1,\dots , x_N)$.
We consider the case where there are on average $t= pN$ phase errors, and consider the limit of large $N$ where the average number of phase errors is held constant.
In this asymptotic limit, $p$ vanishes, and this allows us to side-step the impossibility of achieving a quantum advantage using robust quantum meterology when $p$ is constant and $N$ becomes large \cite{DemkowiczDobrzaski2012}. 

We turn our attention to a different dephasing channel $\mathcal D_\lambda$, where no phase error occurs with probability $\lambda$ and one phase error occurs on each qubit with probability $(1-\lambda)/N$.
Here, $Z_j$ denotes a multi-qubit Pauli that applies $Z$ on the $j$th qubit and acts trivially on the remaining qubits.
Here, the dephased probe state is  
\begin{align}
\mathcal D_\lambda(  |\varphi_2\>\<\varphi_2| ) = 
\lambda |\varphi_2\>\<\varphi_2| 
+ \frac{1-\lambda}{N} 
\sum_{j=1}^N Z_j |\varphi_2\>\<\varphi_2|Z_j .
\end{align}
Our first result pertains a lower bound on the QFI of $\mathcal D_\lambda(  |\varphi_2\>\<\varphi_2| )$, and we present it in the following theorem. 
\begin{theorem}
\label{thm:single-qubit-dephasing-non-asymptotic}
Let $g$ and $n$ be positive integers, and $N=2gn$. Then the QFI of $\mathcal D_\lambda(|\varphi_2\>\<\varphi_2|)$ is at least
\begin{align}
2g^2 n \lambda^2 + 2\lambda(1-\lambda)\frac{g^2(n-1)}{2}  
+ (1-\lambda)^2 \frac{g_2}{N^2}.\notag
\end{align}
where $g_2 = 2(z_{0,2,2}-z_{1,1,2})$
and 
\begin{align}
z_{0,2,2} =& N(N-1) v_2 
\left( (N-2)(N-3)v_4 + (5N-8)v_2 +2 \right) \notag\\
&+ N^2 + N^2(N-1) v_2,\notag\\
z_{1,1,2} =& 
 N(N-1)(N-2)^2   v_3^2 + N(N-1) \notag\\
&  + 2 N(N-1)(N-2) v_3 + N^3/4,
\end{align}
and $v_j$ are as given in Lemma \ref{lem:probe-state-Z-inner-product}.
\end{theorem}
\begin{proof}

Now for every tuple of integers ${\bf i} = (i_1,\dots i_u)$
let $\bar Z_{\bf i} = Z_{i_1} \dots Z_{i_u}$.
Now for every positive integer $u,v$ and $w$, let \begin{align}
    A_{0,v,0} &= 
    \sum_{\substack{ 
     {\bf i} \in \{1,\dots,N\}^v \\
     }}
    \<\varphi_2 | \bar Z_{\bf i} |\varphi_2\>, \notag
    \\
    A_{u,v,0} &= 
    \sum_{\substack{ 
    {\bf i} \in \{1,\dots,N\}^u \\
    {\bf j} \in \{1,\dots,N\}^v \\
      }}
    \<\varphi_2 | \bar Z_{\bf i} |\varphi_2\>
    \<\varphi_2 | \bar Z_{\bf j} |\varphi_2\>, \notag
    \\
    A_{0,v,w} &= 
    \sum_{\substack{ 
    {\bf j} \in \{1,\dots,N\}^v \\
    {\bf k} \in \{1,\dots,N\}^w \\
     }}
    \<\varphi_2 | \bar Z_{\bf k} |\varphi_2\>
    \<\varphi_2 | \bar Z_{\bf k} \bar Z_{\bf j} |\varphi_2\>, \notag
    \\
    A_{u,v,w} &= 
    \sum_{\substack{ 
    {\bf i} \in \{1,\dots,N\}^u \\
    {\bf j} \in \{1,\dots,N\}^v \\
    {\bf k} \in \{1,\dots,N\}^w \\
    }}
    \<\varphi_2 | \bar Z_{\bf k} \bar Z_{\bf i} |\varphi_2\>
    \<\varphi_2 | \bar Z_{\bf k} \bar Z_{\bf j} |\varphi_2\>.   
    \label{eq:z-sums}
\end{align}
Let $\omega= \mathcal D_\lambda(|\varphi_2\>\<\varphi_2|)$.
The lower bound that we use on the QFI is given by $2\tr(\omega^2 H^2)- 2\tr(\omega H \omega H)$.
Since $Z_u$ and $Z_v$ commute, and the trace has a cyclic property, we can 
easily determine that 
$\tr \omega^2 H^2=  \tr ( \sum_{u,v=1}^N \omega^2 Z_u Z_v ).$
From this, we can ascertain that 
\begin{align}
&\tr(\omega ^2 H^2) =  
\lambda^2 A_{0,2,0}
+ \frac{2 \lambda (1-\lambda) }{N} A_{0,2,1} 
+ \frac{(1-\lambda)^2}{N^2} A_{0,2,2}.
\label{eq:dephasing-sshh}
\end{align} 
We can similarly find that 
\begin{align}
&\tr (\omega H \omega H ) 
= \lambda^2 A_{1,1,0}
+ \frac{ 2 \lambda (1-\lambda) }{N} A_{1,1,1} 
+ \frac{ (1-\lambda)^2 }{N^2} A_{1,1,2}. \label{eq:dephasing-shsh}
\end{align}
It remains to show that $A_{i,j,k} = z_{i,j,k},$
where
\begin{align}
z_{0,2,0} =&  N + N(N-1) v_2, \notag
\\
z_{1,1,0} =& N^2/4, \notag
\\
z_{0,2,1} =& 
\frac{ N  +  3 N (N-1)  }{4}  + \frac{N (N-1) (N-2) v_3}{2}, \notag
\\
z_{1,1,1} =& N + 2N(N-1)v_2 + N(N-1)^2 v_2^2.
\end{align}
We show this first using permutation-invariance of the probe state to reorder the Pauli $Z$ operators to only act non-trivially on the first few qubits. 
We then count the number of instances in which the resultant Pauli operator acts non-trivially on zero, one, two, three and four qubits.
Since Dicke states of different weights when multiplied by diagonal matrices remain orthogonal, this allows us to write the sums as 
linear combinations of $Z$-type Dicke inner products that can be calculated using 
Lemma \ref{lem:dicke-inner}, from which the result follows.
\end{proof}
This allows us to readily obtain the following asymptotic result for the dephasing channel $\mathcal D_\lambda$ by taking the limit of $g$ being large in Theorem \ref{thm:single-qubit-dephasing-non-asymptotic}.
\begin{corollary}
\label{coro:dephasing asymptotic}
When $g$ becomes large, the QFI of $\mathcal D_\lambda(|\varphi_2\>\<\varphi_2|)$ divided by $N^2$ is at least
\begin{align}
\frac{\lambda^2}{2 n}  + \lambda(1-\lambda)\frac{n-1}{4n^2}   +(1- \lambda)^2 \frac{(n^3+n-2)}{32 n^4}.\notag
\end{align}
\end{corollary}
When $\lambda = \frac 1 2$, the QFI divided by $N^2$ is at least 
\begin{align}
 \frac{25}{128n}  - \frac{1}{16n^2} + \frac{1}{128n^3} -\frac{1}{64n^4}.
\end{align}
We now compare the performance of our probe state with the GHZ state, after acted upon by $\mathcal D_\lambda$. The worst thing that could happen to a GHZ state when acted on by $\mathcal D_\lambda$ is when $\lambda=(1-\lambda)=\frac 1 2$,   which reduces the GHZ state to a state with zero QFI.
However, Figure \ref{fig:bounds dephasing} which illustrates the results in Thereom \ref{thm:single-qubit-dephasing-non-asymptotic} numerically shows that when $\lambda = 1/2$, our symmetric probe state can still exhibit a quantum advantage.

\begin{figure}[!htb]
\includegraphics[width=1\columnwidth]{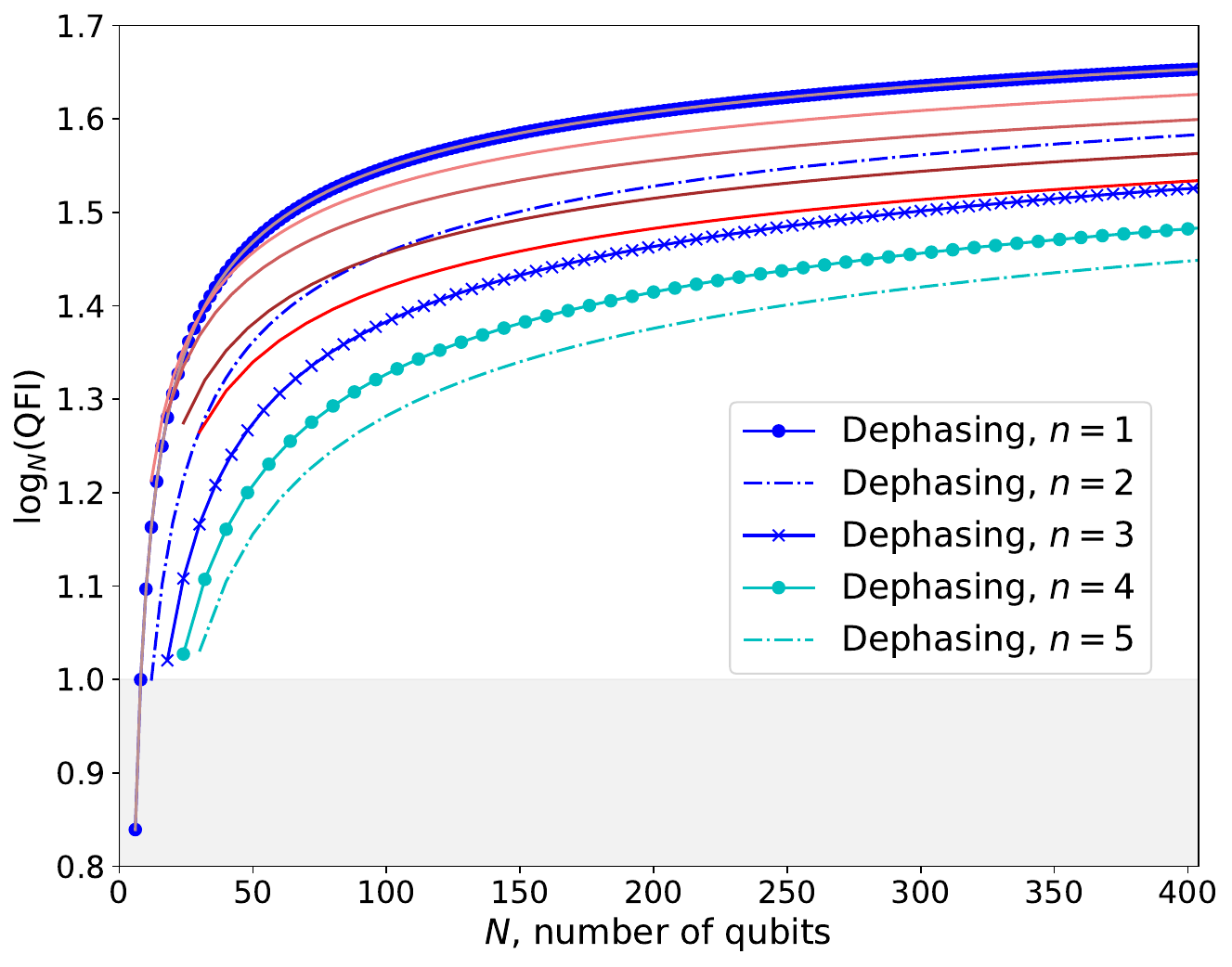}
\caption{Lower bounds for the \textsc{QFI} of $\mathcal D_\lambda(|\varphi_2\>\<\varphi_2|)$ where $|\varphi_2\>$ is the ideal symmetric probe state and $\mathcal D_\lambda$ is a special dephasing channel with $\lambda=\frac 1 2$ that dephases up to 1 qubit. 
When the lines are above 1, there is a quantum advantage in using $|\varphi_2\>$ in the presence of dephasing noise $\mathcal D_\lambda$.
The solid lines from the top give the asymptotic lower bounds on the QFI from Corollary \ref{coro:dephasing asymptotic}. The other lines are non-asymptotic lower bounds on the QFI from Theorem \ref{thm:single-qubit-dephasing-non-asymptotic}.
The asymptotic and non-asymptotic lower bounds agree for $n=1$, and begin to differ for larger $n$.}
\label{fig:bounds dephasing}
\end{figure}
From our lower bound on the QFI of 
$\mathcal D_{\lambda	}( |\varphi_2\>\<\varphi_2|)$, 
we can now derive a lower bound on the QFI of $\sigma$, where $\sigma$ in \eqref{eq:dephasing definition} is the state after encountering an i.i.d. dephasing channel.
We do so by approximating the i.i.d. dephasing channel with the artificial dephasing channel $D_\lambda$ that introduces only zero or one phase error.
This approximation is accurate when i.i.d. dephasing probability is small, so that with overwhelming probability, only zero or one phase errors occur.
\begin{corollary}\label{coro:dephasing-analytic}
Using the notation of Theorem \ref{thm:single-qubit-dephasing-non-asymptotic}
with the definition of the dephased probe state $\sigma$ defined in \eqref{eq:dephasing definition}, the \textsc{QFI} of $\sigma$ is at least
\begin{align}
&(1-p)^{2N-2}\left( 2(1-p)^2 g^2 n  + p(1-p)g^2(n-1)  
+ p^2 \frac{g_2}{N^2}\right) 
\notag\\
&- 4 N^2 t^2 e.\notag
\end{align}
where $t = pN$ is the expected number of dephasing errors and is held constant.
\end{corollary}
\begin{proof}
Recall that $t = pN$.
When $t$ is small, the dephased state $\sigma$ in 
\eqref{eq:dephasing definition} is well approximated by its leading order terms in $j$. To perform a perturbative analysis, we work with an approximant to $\sigma$ given by
\begin{align}
\sigma_k &=    
\sum_{j=0}^{k}  p^j (1-p)^{N-j}
     \sum_{       {\bf x}  \in B_j  } 
    Z_{{\bf x}}     |\varphi_2\>\<\varphi_2|     Z_{{\bf x}}.
\end{align}
Note that $\sigma_N = \sigma$. The approximation error of $\sigma_k$ to $\sigma$ is given by the trace norm of $\epsilon_k$ where 
$\epsilon_k = \sigma - \sigma_k$. Using the linearity of the commutator and the triangle inequality of norms, we have 
\begin{align}
\|  [\sigma , H ] \|_1 =  \|  [\sigma_k , H ] + [\epsilon_k ,H] \|_1  \ge \|  [\sigma_k , H ]\|_1  - \| [\epsilon_k, H] \|_1 .
\end{align}
It follows that 
\begin{align}
\|  [\sigma , H ] \|_1^2   
&\ge \|  [\sigma_k , H ]\|_1^2  - 2\|  [\sigma_k , H ]\|_1 \| [\epsilon_k, H] \|_1  + \| [\epsilon_k, H] \|_1 ^2
\notag\\
&\ge \|  [\sigma_k , H ]\|_1^2  - 2\|  [\sigma_k , H ]\|_1 \| [\epsilon_k, H] \|_1.
\end{align}
Now note that $\|H\|_\infty = N$, and from the triangle inequality, we have $\|\epsilon_k\|_1 \le  \tau_k$,
where
\begin{align}
\tau_k = \sum_{j=k+1}^{N} \binom{N}{j} p^j (1-p)^{N-j}.
\end{align}
From the H\"older inequality, we have 
$\| \sigma_k H \|_1 \le \| \sigma_k\|_1\| H \|_\infty \le N$
and
$\| \epsilon_k H \|_1 \le \| \epsilon_k\|_1\| H \|_\infty \le N\tau_k$.
Similarly,
$\| H \sigma_k \|_1 \le \| \sigma_k\|_1\| H \|_\infty \le N$
and
$\| H \epsilon_k \|_1 \le \| \epsilon_k\|_1\| H \|_\infty \le N \tau_k$.
Hence from the triangle inequality for norms, we have
$\|[\sigma_k , H] \|_1 \le 2N$ and $\|[\epsilon_k , H] \|_1 \le 2N \tau_k$.

Hence we have
\begin{align}
\|  [\sigma , H ] \|_1^2 \ge \|  [\sigma_k , H ]\|_1^2  -8 N^2 \tau_k.
\label{eq:approximant-bound}
\end{align}
Now note the upper bound $\tau_k \le
\sum_{j=k+1}^\infty \binom{N}{j} p^j$. From this, we can see that $\tau_k$ is at most the approximation error of using an order $k$ polynomial to approximate the function $(1+p)^N$. Using the Lagrange form for the remainder term comprising of monomials in $p$ of order at least $k+1$ in the Taylor series expansion of $(1+p)^N$, we know 
$\tau_k \le \frac{p^{k+1}}{(k+1)!}\frac{d^{k+1}}{dp^{k+1}}(1+p)^N 
\le \frac{p^{k+1}}{(k+1)!} N^{k+1} (1+p)^{N-k-1}$. Substituting $t=p/N$, we get
\begin{align}
\tau_k \le
\frac{ t^{k+1}}{(k+1)!}(1+t/N)^{N-k-1}\le 
\frac{ t^{k+1}}{(k+1)!}(1+t/N)^{N}.
\end{align}
For positive $t$, $(1+t/N)^N \le \sum_{j=0}^N \frac{t^j}{j!} \le e^{t}.$
Hence 
\begin{align}
\tau_k \le \frac{t^{k+1}}{(k+1)!}e^t. \label{eq:tail-bound}
\end{align} 
For us, we get $\tau_1 \le t^2 e/ 2$.

Now note that
\begin{align}
\sigma_1 = (1-p)^{N-1} \mathcal D_{1-p}( |\varphi_2\>\<\varphi_2| ).
\end{align}
Hence it follows that 
\begin{align}
\|[\sigma_1, H ]\|_1  = (1-p)^{N-1} 
\|[\mathcal D_{1-p}( |\varphi_2\>\<\varphi_2| ), H]\|_1,
\end{align}
and using Theorem \ref{thm:single-qubit-dephasing-non-asymptotic} with \eqref{eq:approximant-bound} with \eqref{eq:tail-bound}, the lower bound on the QFI for the dephased state $\sigma$ follows.
\end{proof}

From this we can obtain an asymptotic result on the QFI of $\sigma$ in the limit of large $N$, and where the expected number of dephasing errors is held constant.
\begin{corollary}\label{coro:iid-dephasing-asymptotic}
Using the notation of Theorem \ref{thm:single-qubit-dephasing-non-asymptotic}
with the definition of the dephased probe state $\sigma$ defined in \eqref{eq:dephasing definition}, in the limit of large $g$, 
the \textsc{QFI} of $\sigma$ divided by $N^2$ is at least
\begin{align}
 \frac{e^{-2t}}{2 n}  - 4 t^2 e.\notag
\end{align}
where $t$ is the expected number of dephasing errors and is held constant.
\end{corollary}
Corollary \ref{coro:iid-dephasing-asymptotic} implies that our probe state can reproduce the Heisenberg scaling for small number of dephasing errors.
Note that for large enough $t$, the lower bound to the QFI of gnu probe states for i.i.d. dephasing errors in Corollary \ref{coro:dephasing-analytic} and Corollary \ref{coro:iid-dephasing-asymptotic} for the non-asymptotic and asymptotic cases respectively can be negative. In this scenario the lower bounds convey no information, because the QFI must always be at least zero.
Since the function $q(t) = \frac{e^{-2t}}{2 n}  - 4 t^2 e$ is monotone decreasing for positive $t$, $q(t^*) = 0$ for positive $t^*$ only when
$ \frac{1}{8 n e}  = (t^*)^2 e^{2t^*}$, which is when 
\begin{align}
t^* =  W(  \frac{1}{2\sqrt{2 n e}} ), 
\end{align}
where $W$ denotes the principal branch of the Lambert $W$ function. 
When $n=1,2,3,4,5$, we have $t^*=0.17925$, $0.13277$, $0.11082$, $0.0972812$, $0.0878367$ respectively. 
Then for all $t \in [0, t^*]$, the lower bound on the QFI in Corollary \ref{coro:iid-dephasing-asymptotic} will be non-trivial.

We can compare the asymptotic performance of gnu states with respect to i.i.d. dephasing channels using Corollary \ref{coro:iid-dephasing-asymptotic} with noisy GHZ states.
When i.i.d. dephasing noise afflicts GHZ states $(|0\>^{\otimes N} + |1\>^{\otimes N})/\sqrt 2$ where the probability of phase flip per qubit is $p$, 
the probability of the state remaining a GHZ state is 
$p_0 = \frac{( (1-p) - p)^N  + (1-p +p)^N}{2}$ and the probability of the GHZ state acquiring a phase flip is 
$p_1 =  \frac{-( (1-p) - p)^N  + (1-p +p)^N}{2}$. 
The QFI of the noisy GHZ state is $(p_0-p_1)N^2 = N^2 (1-2p)^N$. When $t$ is constant and $p=t/N$, we have $N^2 (1-2p)^N \approx N^2 e^{-2t}$. 
Hence the lower bound of the QFI for gnu states under i.i.d. dephasing noise is lower than the exact QFI of noisy GHZ states under i.i.d. dephasing noise.

\section{Discussions}
In this paper, we study the potential of using explicit symmetric probe states for the purpose of robust metrology. 
In using the terminology robust metrology, we consider a metrological protocol where carefully chosen probe states are allowed to naturally decohere and are subsequently measured.
Our considered noise processes include erasure errors and dephasing errors.
We show that if the probe states lie within the codespace of certain permutation-invariant quantum codes \cite{ouyang2014permutation} which allow for the detection of at least 1 error,
such probe states are useful for robust metrology in the NISQ regime (see Figures \ref{fig:bounds} and \ref{fig:bounds dephasing}). 
We also demonstrate that in the asymptotic regime, our probe states can recover Heisenberg scaling for quantum metrology when the number of erasure errors is very few, and also when the dephasing error becomes increasingly negligible.

To arrive at our lower bounds on the \textsc{QFI}, we rely on explicit structural properties of (1) Dicke states under action of the partial trace, and (2) the connection between Dicke inner products and Krawtchouk polynomials.
The first and second structural properties correspond to erasure and dephasing errors respectively. 
In both cases, we reduce the problem to that of performing binomial-type summations of the form $\sum_{j=0}^n \binom n j j^x$ in the asymptotic limit, when $g$ becomes very large and $x<n$. 
Leveraging on this, we obtain compact analytical lower bounds for the \textsc{QFI} in the case of robust metrology.

For erasure errors, from Theorem 2, we can see that the QFI decreases exponentially with $n$ when $g,n>t$. 
For the dephasing channel $\mathcal D_\lambda$, since $2gn=N$ is constant, the lower bound in the QFI in Theorem 7 decreases linearly with $n$.
For the i.i.d. dephasing channel in Corollary \ref{coro:iid-dephasing-asymptotic}, we also see that the lower bound on the QFI decreases linearly with $n$.

From Figures \ref{fig:bounds} and \ref{fig:bounds dephasing}, when there are at least 50 qubits, there is a discernable quantum advantage in using our proposed probe state for robust metrology.
For example, if only 1 out of 50 qubits is erased, the corrupted probe state $|\varphi_1\>$ can nonetheless attain a \textsc{QFI} of at least $N^{1.4}$ which is about 5 times larger than the baseline \textsc{CFI} of $N$. 

Our results complement the literature on robust quantum metrology \cite{PRX-random-states} and quantum metrology with error correction \cite{kessler2014quantum,dur2014improved,arrad2014increasing,unden2016quantum,matsuzaki2017magnetic,layden2019ancilla,gorecki2019quantum,PRXQuantum.2.010343} by considering symmetric probe states that support quantum error correction.
Because our probe states inherit the quantum error correction (\textsc{QEC}) capability of the permutation-invariant quantum codes illustrated in Ref.~\cite{ouyang2014permutation}, we know that active quantum error correction can always be employed on our probe states to further amplify the \textsc{QFI}.
We expect that the \textsc{QFI} can be further enhanced with full blown \textsc{QEC}, and this remains to be explored.

Analyzing the performance of these symmetric probe states from gnu codes with respect to a number of errors that grows with the system size is an important problem. However this is beyond the scope of current manuscript, and we leave this interesting question for future work.
Deriving a fully fleshed out practical proposal based on our scheme for robust metrology remains an open challenge. 
While our probe states can be prepared in the ultrastrong coupling regime of superconducting charge qubits \cite{init-picode-2019-PRA} or using geometric phase gates \cite{johnsson2020geometric}, questions of state preparation, decoding \cite{ouyang2021permutation} and state readout in a broad range of other architectures remain to be fully addressed. 

In summary, we prove that explicit symmetric probe states can give rise to a quantum advantage in the NISQ regime spite of being mildly corrupted.
Moreover, in the asymptotic limit, we show that it is possible for our proposed symmetric probe states to attain a Heisenberg scaling even if they suffer a non-trivial amount of noise.
This paves the way towards exploring the effectiveness of our symmetric probe states against different types of noise processes, such as amplitude damping noise and depolarizing noise.
It is also interesting to further study the extent in which our symmetric probe states can remain effective when multiple parameters are to be simultaneously estimated \cite{sidhu2019tight}.
More fundamentally, it is interesting to unravel the necessary and sufficient conditions for robust metrology, at least in the case of symmetric qubit states.
We believe that the strength of our symmetric probe states in robust metrology is that the noise processes considered do not take the probes state to orthogonal states. Our probe states, arising from permutation-invariant quantum codes, exhibit a highly non-additive behavior, and it might be possible that this leads to their utility in quantum metrology. Indeed, we are tempted to conjecture that any permutation-invariant code that can detect at least one error is good for quantum metrology with respect to erasure errors.

\subsection{Verification of quantum metrology}

In a cryptographic setting, we want to be sure that the probe state that we intend to use for quantum metrology is what it is supposed to be. If the probe state has maliciously tampered with, the final value of the estimated signal need not reflect the true value of the signal.

To mitigate this type of malicious attack, we can employ verification techniques prior to performing the quantum metrology task. Verification of quantum states is done by performing quantum tomography on additional quantum states \cite{takeuchi2018, pallister2018, zhu2019} and is generally done to add a measure of cryptographic security if quantum states are provided from an untrusted party \cite{zhu2019}, or simply to ensure that the quantum state preperation is being done correctly \cite{flammia2011}.

The efficiency of a verification protocol is based on a cryptographic quantity known as soundness, which quantifies the ability of a malicious adversary tampering to tamper with the quantum states and remain undetected. Therefore, an ideal verification protocol requires very few additional resources to achieve a low soundness value. Most verification protocols are designed for a specific class of quantum states, as they take advantage of specific constraints or symmetries unique to said class, for example graph states \cite{zhu2019GraphStates, markham2020} and Dicke states \cite{liu2019}. There are general protocols to verify arbitrary quantum states \cite{takeuchi2018}, however for most quantum states (including the symmetric probe states), these protocols would require exponentially more resources to achieve the same level of soundness achieved by the graph state or Dicke state verification protocols.

Unfortunately, we do not know of any efficient verification protocols that use local measurements and operate with a polynomial number of resources. 
In lieu of considering verification protocols that use non-local measurements, we could consider verifying a symmetric state which well approximates our gnu probe state that furthermore is also known to admit an efficient verification procedure. In this scenario, we note that for $u=1$, the gnu probe state gets closer to the half-Dicke state $|D^N_{N/2}\>$ as $n$ increases. Thus, in practice one can utilize the Dicke state verification protocol \cite{liu2019} for the symmetric probe state as a proxy for verifying the gnu probe state. This can be done with fewer additional resources than a more general protocol, however it comes with the caveats of i) not truly verifying the quantum state, and ii) possible failure of the verification protocol.

We hope that in future work, further research can be done on the verification of symmetric probe states for the robust quantum metrology.

\section*{Acknowledgements}

YO acknowledges support from the Singapore National Research Foundation under NRF Award NRF-NRFF2013-01, the U.S. Air Force Office of Scientific Research under AOARD grant FA2386-18-1-4003, and the Singapore Ministry of Education. 
YO acknowledges support from EPSRC (Grant No. EP/M024261/1).
and the QCDA project (EP/R043825/1) which has received funding from the QuantERA ERANET Cofund in Quantum Technologies implemented within the European Union’s Horizon 2020 Programme. NS and DM gratefully acknowledge support from the ANR through the ANR-17-CE24-0035 VanQuTe project.

\bibliography{robustmetro}{}
\bibliographystyle{ieeetr}

\appendix

\section{Commutation identities}
For Hermitian matrices $A$ and $B$, we have that
\begin{align}
&\tr( [A,B]^\dagger [A,B])\notag\\
=&
\tr ( (AB-BA)^\dagger (AB-BA)  ) 
\notag\\
=&
\tr (  (BA-AB) (AB-BA)  ) 
\notag\\
=&
\tr (  (BAAB-BABA-ABAB +ABBA)  ) 
\notag\\
=&
2 \tr (A^2 B^2) - 2 \tr( (AB)^2 ) .\label{eq:norm2-[A,B]}
\end{align}
Also note that for Hermitian matrices $A,B,C$, we have
\begin{align}
&\tr( [A,B]^\dagger [C,B])\notag\\
=&
\tr ( (AB-BA)^\dagger (CB-BC)  ) 
\notag\\
=&
\tr (  (BA-AB) (CB-BC)  ) 
\notag\\
=&
\tr (  (BACB-BABC-ABCB +ABBC)  ) 
\notag\\
=&
\tr ((AC+CA)B^2) - 2 \tr( AB CB ) .\label{eq:norm2-[ABC]}
\end{align}

\section{Technical proofs}
\label{app:A}

\begin{proof}[Proof of Theorem \ref{thm:erasures-exact}]
We proceed to evaluate both $ \tr (\rho^2 H^2)$
and $  \tr( \rho H \rho H ) $. 
Using the linearity of the trace and the fact that $Z_i^2$ is the identity operator,
we get $ \tr (\rho^2 H^2)   =    \sum_{i,j=1}^{N-t}   \tr (\rho^2 Z_i Z_j )  
    = 
       \sum_{i=1}^{N-t} \tr (\rho^2  Z_i^2)    +        \sum_{ i\neq j} \tr (\rho^2  Z_i Z_j) $. 
The permutation-invariance of the partial trace of the probe state then gives
\begin{align}
    \tr (\rho^2 H^2)=  (N-t)    \tr (\rho^2 )     +   2  \binom{ N-t} {2}    \tr (\rho^2  Z_1 Z_2) . \label{eq:1}
\end{align} 
From the representation of the partial trace of our symmetric probe state $\rho$ given in Lemma \ref{lem:erasure-representation}, we readily obtain from the orthogonality relationship (\ref{eq:psi-ortho})
and orthogonality of Dicke states that
\begin{align}
  \rho^2 &= 
\sum_{u=0}^t \binom t u ^2      |\theta_u\>\<\theta_u|  \sum_{j=1} ^{n-1} a_{j,u}
+
W_1 + W_2 + W_3.
\end{align}
where
\begin{align}
2^n W_1 &=    |H^{N-t}_{0}\>\<H^{N-t}_{0}| \<\theta_0|\theta_0\>
 +
    |H^{N-t}_{N-t}\>\<H^{N-t}_{N-t}| \<\theta_t|\theta_t\>
    \notag\\
    &\quad
+       |\theta_0\>\<\theta_{0}|   +   |\theta_t\>\<\theta_{t}|      
\notag\\
2^{3n/2} W_2 &=  |\theta_0\>\<H^{N-t}_{0}|    +       |\theta_t\>\<H^{N-t}_{gn-t}|   \notag\\
  &\quad +
  |H^{N-t}_{0}\>\<  |\theta_0|   +  |H^{N-t}_{gn-t}\>\<  \theta_t|
\notag\\ 
W_3 &=
\frac{1}{2^{2n}} \left( (|0\>\<0|)^{\otimes N-t} + (|1\>\<1|)^{\otimes N-t} \right). 
\end{align}
It readily follows that
\begin{align}
\tr (\rho^2)  &=  \sum_{u=0}^t \binom t u^2 \sum_{j,k=1}^{n-1} a_{j,u}a_{k,u}
+  \sum_{j\in \{0,t\}} \frac{2 \<\theta_j|\theta_j\>}{2^n} + \frac{2}{2^{2n}},  \label{eq:2}
\end{align}
and
\begin{align}
\tr(\rho^2 Z_1 Z_2 ) &=  
\sum_{u=0}^t \binom t u^2 \sum_{j=1}^{n-1} a_{j,u} \<\theta_u|Z_1 Z_2|\theta_u\>
\notag\\
&+ 
 \sum_{j\in \{0,t\}} 
\frac{ \<\theta_j|\theta_j\>  + \<\theta_j| Z_1 Z_2|\theta_j\>}{2^n} + \frac{2}{2^{2n}}.
\end{align}
Next, we use (\ref{eq:dickeZ2})  to get
\begin{align}
&     \<\theta_u |Z_1 Z_2 |\theta_u\> \notag\\
    =&
    \sum_{j,k=1}^{n-1}
        \frac{\sqrt{\binom n j \binom n k}}
        {\sqrt{\binom N{gj} \binom N{gk} } 2^{n} } 
    \< H^{N-t}_{gj-u} | Z_1 Z_2 | H^{N-t}_{gk-u} \>
    \notag\\
    =&
    \sum_{j=1}^{n-1}
        \frac{\binom n j }
        {\binom N{gj}  2^{n} } 
    \< H^{N-t}_{gj-u} | Z_1 Z_2 | H^{N-t}_{gj-u} \>
    \notag\\    
        = &
    \sum_{j=1}^{n-1}
    a_{j,u} \left( 1 - b_{j,u}
    \right)\label{eq:3}.
\end{align}
Using  
$\tr (\rho^2 H^2) = \sum_{j=1}^{N-t} \tr(\rho^2 ) + 
\sum_{j\neq k} \tr(\rho^2 Z_i Z_j ) $, we use   (\ref{eq:1}), (\ref{eq:2}) and (\ref{eq:3}) 
to get
\begin{align}
 &\tr (\rho^2 H^2) \notag\\
    &=  
  \frac{  (N-t)^2}{2^{2n-1}} 
  +
 \frac{ (N-t)^2}{2^n}
  \sum_{j=1}^{n-1} \left( a_{j,0}(2-b_{j,0}) + a_{j,t}(2-b_{j,t}) \right) \notag\\
  &+
\frac{ N-t}{2^n}
  \sum_{j=1}^{n-1} \left( a_{j,0} b_{j,0} + a_{j,t} b_{j,t} \right)   
    \notag\\
&   +
(N-t)^2  \sum_{u=0}^t \binom t u ^ 2 \sum_{j=1}^{n-1}   \sum_{k=1}^{n-1}
           a_{j,u}     a_{k,u} 
     \left( 1 - b_{k,u} + \frac{b_{k,u}}{N-t}	\right)
    \label{eq:4}.
\end{align}
Now we proceed to evaluate $\tr(\rho H \rho H)$.
Orthogonality and permutation-invariance of the states $|\theta_u\>$ implies that  
\begin{align}
&
 \frac{ \rho	H \rho }{N-t}   \notag\\
 = &
\sum_{u=0}^t \binom t u ^2      |\theta_u\>\<\theta_u|   \<\theta_u| Z_1  |\theta_u\>
\notag\\
&+
\frac{1}{2^n} \left(   |H^{N-t}_{0}\>\<H^{N-t}_{0}| \<\theta_0|Z_1|\theta_0\> 
 +
    |H^{N-t}_{N-t}\>\<H^{N-t}_{N-t}| \<\theta_t|Z_1|\theta_t\> 
+       |\theta_0\>\<\theta_{0}|    -   |\theta_t\>\<\theta_{t}|   \right)    
\notag\\
&+
\frac{1}{2^{3n/2}}\left(  |\theta_0\>\<H^{N-t}_{0}|    -    |\theta_t\>\<H^{N-t}_{gn-t}|    +  |H^{N-t}_{0}\>\<  |\theta_0|   -  |H^{N-t}_{gn-t}\>\<  \theta_t|\right) 
+
\frac{1}{2^{2n}} \left( (|0\>\<0|)^{\otimes N-t} - (|1\>\<1|)^{\otimes N-t} \right)
\notag\\
&+
\frac{1}{2^{n/2}}  \left(
\<\theta_0| Z_1  |\theta_0\> (   |H^{N-t}_{0}\> \<\theta_{0}|  +   |\theta_{0}\> \< H^{N-t}_{0}|  )
+ 
\<\theta_t| Z_1  |\theta_t\> ( |H^{N-t}_{gn-t}\> \<\theta_t|   + |\theta_t\>\< H^{N-t}_{gn-t}| \right).
\end{align}
Then we get
\begin{align}
& \frac{  \tr(\rho H \rho H) }{(N-t)^2} \notag\\
    =&
\sum_{u=0}^t \binom t u ^2      \<\theta_u| Z_1  |\theta_u\>^2
+
\frac{2}{2^n}   \left(    \<\theta_0|Z_1|\theta_0\>  - \<\theta_t|Z_1|\theta_t\>  \right)    
+
\frac{2}{2^{2n}}  
 .\label{eq:5}
\end{align}
We now find using (\ref{eq:dickeZ1}) that
\begin{align}
     \<\theta_u |Z_1 |\theta_u\>
    &=
    \sum_{j=1}^{n-1}
        \frac{\binom n j }
        {\binom N{gj}  2^{n} } 
    \< H^{N-t}_{gj-u} | Z_1 | H^{N-t}_{gj-u} \>
    \notag\\    
    &= 
    \sum_{j=1}^{n-1}
        \frac{\binom n j }
        {\binom N{gj}  2^{n} } 
    \binom{N-t}{gj-u} \left( 1 - \frac{2(gj-u)}{N-t}
    \right)
    \notag\\
&=
    \sum_{j=1}^{n-1} a_{j,u} \left( 1 - c_{j,u}   \right).
\end{align}
Hence we have
\begin{align}
&  \frac{  \tr(\rho	H \rho H) }{(N-t)^2} \notag\\
=&
\sum_{u=0}^t \binom t u ^2   
\sum_{j=1}^{n-1}
\sum_{k=1}^{n-1}
     a_{j,u} a_{k,u} (1-c_{j,u})(1-c_{k,u}) \notag\\
&+
\frac{2}{2^n}
\sum_{j=1}^{n-1}
  \left(   a_{j,0} (1-c_{j,0}) -a_{j,t} (1-c_{j,t})   \right)    
+
\frac{2}{2^{2n}}   .
 \label{eq:tr rho h rho h}
\end{align}
From (\ref{eq:4}) and (\ref{eq:tr rho h rho h}) we get
\begin{align}
&  
    \frac{ \tr(\rho^2H^2) - \tr(\rho	H \rho H) }{(N-t)^2}
    \notag\\ 
=&
 \frac{1}{2^n}
  \sum_{j=1}^{n-1} \left(
   a_{j,0}(-b_{j,0}+2c_{j,0}) + a_{j,t}(4-b_{j,t}-2c_{j,t})  
   \right)\notag\\
&  +
\frac{1}{2^n(N-t)}
  \sum_{j=1}^{n-1} \left( a_{j,0} b_{j,0} + a_{j,t} b_{j,t} \right)   
    \notag\\
&   +
  \sum_{u=0}^t \binom t u ^ 2 \sum_{j=1}^{n-1}   \sum_{k=1}^{n-1}
           a_{j,u}     a_{k,u} 
     \left(c_{j,u} + c_{k,u} - b_{k,u} \right.\notag\\
     &\left. \quad\quad\quad\quad\quad\quad\quad\quad\quad\quad\quad\quad + \frac{b_{k,u}}{N-t}   -  c_{j,u} c_{k,u}  	\right).
\end{align}
Rearranging the terms gives the result.
\end{proof}
\end{document}